\documentclass{llncs}

\usepackage{boxedminipage}     % for surrounding boxes in figures
\usepackage{xspace}     
\usepackage{pstricks}
\usepackage{pst-node}
\usepackage{pdftricks}
\usepackage{subfigure}
\def\ratetwo#1#2{\mathbf{q}(#1,#2)}
\def\ratethree#1#2#3{\mathbf{q}_{#3}(#1,#2)}
\def\rate#1#2{\@ifnextchar{\bgroup}{\ratethree{#1}{#2}}{\ratetwo{#1}{#2}}}
\makeatother
\makeatletter
\def\rateactthree#1#2#3{\mathbf{q}(#1,#3,#2)}
\def\rateactfour#1#2#3#4{\mathbf{q}_{#4}(#1,#3,#2)}
\def\rateact#1#2#3{\@ifnextchar{\bgroup}{\rateactfour{#1}{#2}{#3}}{\rateactthree{#1}{#2}{#3}}}
\makeatother

\def \Nil{0}
\newcommand{ \Idf}{\mathsf D}
\newcommand{\Proctwo}{\mbox{$\mathcal L$}} 
\newcommand{\stsp}[1]{{\sf S_{#1}}}
\newcommand{\Proc}{\mbox{$\mathcal P$}} 
\newcommand{\SetDerivatives}{\mbox{$\mathcal S$}} 
\newcommand{\Expr}{\mathsf E}
\newcommand{\newExpr}[1]{{\mathsf #1}}

\newcommand{\Aut}{\mathsf M}
\newcommand{\IAut}{\mathsf C}

\usepackage{latexsym}
\usepackage{amssymb,amsmath}
\usepackage{graphicx}
\newcommand{\old}[1]{}
\newcommand{\CTMC}{{\sf CTMC}}

\newcommand{\statespace}{\mbox{\sl S}}
\newcommand{\Real}{{\sf I\kern-0.14emR}} 
 
\newcommand{\generator}{\mathbf{Q}}

\newcommand{\clos}[3]{#1_{[ #2 \leftarrow #3]}} 
\newcommand{\proc}[2]{{\mathsf P_{#1,#2}}}

\newcommand{\newproc}[3]{{\mathsf P^{#3}_{#1#2}}}

\newcommand{\myset}[1]{ \{ #1\}}

\newcommand{\transition}[3]{#1 \xrightarrow{#3} #2}

\makeatletter
 \def \RuleusingThree[#1]#2#3{\prooftree #2\justifies#3 \using{(#1)}
 \endprooftree}
 \def \Ruleusing[#1]#2{\@ifnextchar\bgroup {\RuleusingThree[#1]{#2}}
         {\prooftree \justifies #2 \using{(#1)} \endprooftree}}
 \def \RulenotusingTwo#1#2{\prooftree #1 \justifies #2 \endprooftree}
 \def \Rulenotusing#1{\@ifnextchar\bgroup {\RulenotusingTwo{#1}}
         {\prooftree \justifies #1 \endprooftree}}
 \def \Inf{\proofrulebaseline=2.2ex
         \abovedisplayskip12pt\belowdisplayskip12pt
         \abovedisplayshortskip8pt\belowdisplayshortskip8pt
         \@ifnextchar[{\Ruleusing}{\Rulenotusing}}%]
\makeatother

\newcommand{\uniqueLabels}[1]{{\mathcal U}(#1)}

\newcommand{\Defeq}{\stackrel{\mathit{df}}{=}}

\newcommand{\evolves}[1]{\stackrel{\mbox{\tiny{$\scriptscriptstyle{#1}$}}}{\lra}}

\newcommand{\lra}{\longrightarrow}

\newcommand{\Var}{\mathrm{Var}}

\newcommand{\MA}{\Aut}

\newcommand{\set}[1]{\mathcal{#1}}

\newcommand{\coop}[1]{\bigotimes_{#1}}
\newcommand{\pasact}{\mathcal{ P}}
\newcommand{\actact}{\mathcal{ A}}
\newcommand{\actypepas}[1]{\mathcal{P}(#1)}
\newcommand{\actypeact}[1]{\mathcal{A}(#1)}
\newcommand{\Act}{\mbox{\it Act}}

\newcommand{\ruleOne}[2]{\arraycolsep=14pt\renewcommand{\arraystretch}{1.75}%
\begin{array}[c]{c}#1\\
        \hline\multicolumn{1}{c}{#2}
\end{array}}

\begin{document}

\title{Operational semantics for product-form solution}
\author{Maria Grazia Vigliotti}
\institute{Department of Computing, Imperial College London,\\ 180 Queen's Gate,  London SW7 2BZ, UK \\
\email{maria.vigliotti@imperial.ac.uk}
}

%\institution{Department of Computing, Imperial College London }
\maketitle
\thispagestyle{empty}

\begin{abstract}
In  this paper we present product-form solutions from the point of view of stochastic process algebra. In previous work \cite{Marin-Vigliotti10bis} we have shown how to  derive 
product-form solutions for a formalism called Labelled Markov Automata (LMA). LMA are very useful as their relation  with  the Continuous Time Markov Chains is very direct. The disadvantage  of using LMA is that the proofs of properties  are  cumbersome. In fact,  in LMA it is not possible  to use the inductive structure of the language in a proof.   In this paper we consider a simple stochastic process algebra that has the great advantage of  simplifying  the  proofs. This simple language has been  inspired by PEPA \cite{Hillston94}, however, 
  detailed analysis of  the semantics of  cooperation  will  show the differences between  the two formalisms. It will  also  be shown that  the semantics of the cooperation in process algebra influences the correctness of the derivation of the product-form solutions. 
\end{abstract}

\section{Introduction}

In  this paper we present product-form solutions from  the point of view of stochastic process algebra. The main  motivation for this work is twofold: on  one side, our formalisation clarifies the  basic mechanisms that govern  product-form solutions in Continuous Time Markov Chains (\CTMC s), on the other side, we can generalise the  notion  of product-form solutions beyond  queuing theory. Product-form solutions are efficient solutions for  stationary distributions in {\CTMC}s in general,  while so far product-form solutions have been  studied mostly in the area of performance evaluation/ queuing theory.
The work presented here is an extension of  previous  work \cite{Marin-Vigliotti10bis} where  we have shown how to  derive 
product-form solutions for a formalism called Labelled Markov Automata (LMA).  In very simple terms, LMA  describe  the  state space of  {\CTMC}s as labelled graph  decorated with transition  rates.  LMA are equipped with a basic mechanism  to build  complex {\CTMC}s.  LMA have proved very useful in helping to  understand basic  mechanisms that govern  product-form solutions, and in providing a very elegant  proof of the theorem  GRCAT \cite{Marin-Vigliotti10bis}. However,  the disadvantage  in using LMA is that the proofs,  even for  simple properties,  are  cumbersome.  In LMA  it is not possible  to use the inductive structure of the language in a proof.  

 In this paper we  improve on  previous work \cite{Marin-Vigliotti10bis} by considering   a simple stochastic process algebra  (SSPA), which  preserves the  semantics of cooperation of LMA. 
   This simple language has been  inspired by PEPA \cite{Hillston94}, however 
  detailed analysis of  the semantics of  cooperation  will  show the differences between  the two formalisms.

In this paper we  investigate the general principle that determines product-form solutions in  {\CTMC}s, and  we   show that  the semantics of cooperation is  crucial  for  the correct derivation  of product-form  solutions. 
We will  introduce a simple language equipped with a rather unique, and possibly  counterintuitive semantics, which  guarantees the existence of product-form solutions.  
We will argue that  the semantics presented here,  is precisely what is needed to  model rigorously  product-form solutions for {\CTMC}s.
%In SSPA we will  introduce a new operator, the {\em  closure} which  will allow an elegant definition of product-form solutions.   %The paper will also   show  that  the semantics of the cooperation in process algebra influences the correctness of the derivation of the product-form solutions. 
We shall also  consider a biological  example to show an  interesting application of product-form solutions to a context different from queuing theory.
 
 \section{Related work}
 There is vast literature  on the topic of product-form solutions and process algebra, and on the formalisation of  the intrinsic mechanisms  that determine product-form solution 
 \cite{sereno:pfpepa,hillston:product,harrison:exploiting,balbo:pfgspn,harrison:rcat,coleman:product.form,balbo:pfgspn,balbo:relationsbcmp-pfspn}.  On the relationship  between  process algebra and product-form solutions, Hillston, Thomas and Clark,   played a major role \cite{Hillston-rep-pf98,hillston:product,harrison:exploiting,prodformsoln,Clark2002,sereno:pfpepa}. The common denominator in  these papers is the use of PEPA 
 to model  processes that are known  to enjoy  product-solution, and to extract, via PEPA,  the modular properties of such  processes.   This body  of  work  has demonstrated  to the community  the modelling  power  of PEPA. It was shown   in  \cite{harrison:exploiting} that quasi-reversible structures can be  modelled in PEPA, together with a large variety of product-from solutions. We differ  from the work carried out in PEPA, as  our goal  is not to define  'yet another stochastic process algebra' to  model  product-form solutions, but to  design  a language and a semantics  to  describe only the {\CTMC}s that enjoy product-form solutions.  With  our formalism it is relatively easy  to find new  product-form solutions for {\CTMC}s.  This was not achieved in  previous work. 
 
 Another formalism that has been extensively used   is the Generalised Stochastic Petri Nets (GSPN) \cite{balbo:pfgspn,coleman:product.form,balbo:pfgspn,balbo:relationsbcmp-pfspn}- to cite a few articles. The emphasis is to understand which GSPN enjoy  product-form solutions, and  what  conditions on GSPN are necessary to   yield product-form solution \cite{coleman:product.form}.  We  differ from the  work on the GSPN as we use process algebra, and also  because of  the generality  of our results.   In  this paper, and  in previous work,   the effort has been directed in  defining a set of  sufficient  conditions that guarantee product-form solutions for  time-homogenous {\CTMC}s.
Finally, it must be mentioned that similar efforts have been carried out by the community in  performance \cite{Kelly'79,Chao-Masa-Pinedo,Robertazzi}.

 The class of product-form solutions considered by \cite{Robertazzi} is rather limited, while a  true advancement was made by  \cite{Kelly'79,Chao-Masa-Pinedo} with the notions of {\em quasi-reversibility}. Quasi-reversibility was introduced by Kelly \cite{Kelly'79}, and used only  on the context of queuing networks.  In  \cite{Chao-Masa-Pinedo},   great efforts were successfully made  to show generality  and robustness  of quasi-reversibility.   Nearly all  product-form solutions  known in queueing networks are  derived using  quasi-reversibility.  It  was proved in  \cite{Chao-Masa-Pinedo} that   quasi-reversibility is  a sufficient  condition for product-form solutions.  The conditions of Theorem \ref{grcat} can  be seen as  formalisation  of quasi-reversibility. The main  difference between  Theorem \ref{grcat}  and quasi-reversibility lies in the formalisation of the  `connection' of {\CTMC}s. The way  in which queues are connected together is expressed in  natural language \cite{Chao-Masa-Pinedo}.   This is the main weakness of the   work.  Since it is not clear how to `connect' {\CTMC}s together, only networks of queues are considered. Understanding of how to connect queues together   is clear  in the community  of performance evaluation. Our work goes  further as it specifies,  in a rigorous way, the `connection' or, better, the cooperation among {\CTMC}s. % allowing to  derive product-form solution for time-homogenous      {\CTMC}s in general.
  Note that we have {\em only} sufficient  conditions for  product-form solutions, not necessary conditions. As consequence, there are  product-form solutions that we cannot characterise, for example \cite{boucherie:characterisation}. We leave for future work formal  development to deal  with  such  product-form solutions.
% The rest of the paper is organise as follows: in section \ref{LMA} we introduce the language of SSPA, and in ... 
 
 %A lot of work  has been  done on Generalised Petri nets \cite{}
\section{A simple language}\label{LMA}
In this section  we introduce a  simple stochastic process algebra, SSPA.  The main  motivation  to  introduce such a formalism is to verify that the conditions for  product-form solutions can  be modelled in a language.  SSPA is defined in a rather unusual way, but follows in the spirit ideas that were discussed in \cite{Hillston-rep-pf98,hillston:product}. We initially define  simple processes. These are composed essentially  by choice and by recursion.% These are the processes  that were meant to be in  vector form in \cite{Hillston-rep-pf98,hillston:product}. 
 Similarly to PEPA, simple processes may  or may not characterise a {\CTMC}.  Some simple process are {\em  incomplete}, according to PEPA terminology,  in the sense that some transitions lack the information about the rate. Such  information  can be inserted via  a  new operator: the  {\em closure}. %Such an operator is key for the definition  of product-form solution in theorem \ref{grcat}.
A second layer of processes is  defined, as cooperation of simple processes.  The operator for cooperation has been  inspired by PEPA, but differently form PEPA is an n-nary -operator, like choice. Such operator, differently form  PEPA-bow, cannot be expressed as multiple composition of the binary association. 

 %%% 
To formally define the language we assume the existence of a set of variables $\Var $, and a set of actions $\Act$ on which   the letter $a,b,c \ldots$ range over it.
 \begin{definition}[Simple processes]\label{def:sp}
 The set of  {\em simple process}, $\Proc$, is is given by the following syntax:
\[\begin{array}{rcl}
\Expr&::=& \Nil ~\mid~  \Idf  \\ 
\Aut &::=  &\sum_{i \in I}(a_i,r_i). \Expr _i ~\mid~  \clos{\Aut}{a}{\lambda}    
%A& :: =& \coop{L}(E_1, \ldots E_n)
\end{array}\]
where $I$ is a finite set of indexes.

  \end{definition}

For clarity  in the notation we    use   the  greek letters $\lambda,  \mu, \ldots $ range over the  set of positive real numbers $\Real^+$, the letters $x,y,x, \ldots $  range over $\Var $ and  the letter $r$ ranges over $\Real^+ \cup \Var $.
When writing a variable in processes, we use a subscript that refers to the label. For example we would  write $(a,y_a). \Expr $ but not 
$(a,y_b). \Expr $. %The variable play an  important role in the derivation of the {\CTMC}, and this notation  will  be very useful.
A simple processes  stand for  the building blocks which are ultimately used to  compose complex {\CTMC}s. %Note, that different from  other definition  of Markov automata \cite{}
% transitions of a LMA can  be decoreated with  positive real  number or with some more esoteric `variable'.  We  leave the explanation  of the role of the variables for later.  For the moment, it suffice to  say that transitions decorated with variables have a similar role to the passive transitions in PEPA \cite{}.
Nil, written $\Nil$,  is the empty  process; $\Idf$ is the symbol for  the  identifier. Identifiers are equipped with   {\em  identifier  equations} such as  $\Idf \Defeq \Aut$. The {\em  choice}, $\sum_{i \in I}(a,r). \Expr _i$,  represents the standard selection of one of the possible transitions, and finally there is a new operator {\em closure} $\clos{\Aut}{a}{\lambda} $.   This operator  transforms all transitions  labelled with the pair  $a$  and   a variable  into  transitions labelled with the pair  $a$ and the real number $\lambda$.
 The role of closure will become clear in the later development of the paper.
The grammar of simple process  aims to define the transition  graph of a labelled {\CTMC}, but not all  transition  graphs derived from  this grammar are {\CTMC}s due to the presence of variables. Examples of this kind can  be seen  in Fig. \ref{pas-gif}.  

\begin{figure}[h!]
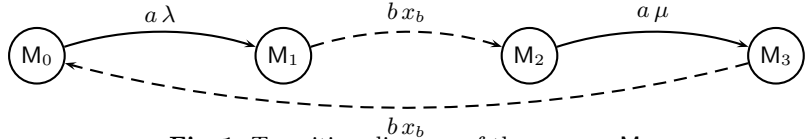

%\begin{minipage}{0.90\textwidth}
\centering
$\psmatrix[mnode=circle,colsep=2.5]  \MA_0& \MA_1 &\MA_2  &\MA_3 \\
\endpsmatrix
\psset{shortput=nab,arrows=->,labelsep=3pt} \small
 \ncarc[arcangle=18]{1,1}{1,2}^{a\, \lambda}
\ncarc[arcangle=18,linestyle=dashed]{1,2}{1,3}^{b \, x_b}
\ncarc[arcangle=18]{1,3}{1,4}^{a\,  \mu }
\ncarc[arcangle=15,linestyle=dashed]{1,4}{1,1}^{b \, x_b} 
$
\vspace*{0.3cm}
%\end{minipage}
\caption{Transition  diagram  of the process $\MA_0$}
\label{pas-gif}
\end{figure}
\begin{figure}[h]
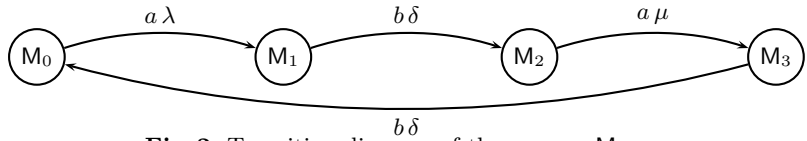

%\begin{minipage}{0.90\textwidth}
\centering
$\psmatrix[mnode=circle,colsep=2.5]  \MA_0& \MA_1&\MA_2&\MA_3 \\
\endpsmatrix
\psset{shortput=nab,arrows=->,labelsep=3pt} \small
 \ncarc[arcangle=18]{1,1}{1,2}^{a\, \lambda}
\ncarc[arcangle=18]{1,2}{1,3}^{b\, \delta}
\ncarc[arcangle=18]{1,3}{1,4}^{a\,  \mu }
\ncarc[arcangle=15]{1,4}{1,1}^{b \, \delta} 
$
\vspace*{0.3cm}
%\end{minipage}
\caption{Transition  diagram  of the process $\MA_{0[b\leftarrow \delta ]}$} %$\clos{\MA_0}{b}{\delta}$}
\label{act-gif}
\end{figure}

%The behaviour of the closure operator will become clear in the presentation of the semantics.
 Informally,  we can say that the closure operator would transform a transition  graph  of the simple process $\MA_0\Defeq (a,\lambda).(b,x_b). (a,\mu).(b,x_b).\MA_0$ as in Fig.  \ref{pas-gif} as one in  Fig.\ref{act-gif} for the simple process $\MA_{0[b\leftarrow \delta ]}$.

If all transitions of a process are  decorated with a  real number, as in  in Fig.  \ref{act-gif}, then  the underlying model  description  is a time-homogenous {\CTMC}, similar to PEPA.  To formally define how to derive the {\CTMC} of a given simple process, we need to give a formal semantics to SSPA via {\em labelled transition system}.
\begin{definition} A   labelled transition system for simple processes  written $ \rightarrow:\Proc \times (\Act
\times \Real^+\cup \Var)
\times \Proc$ % written $ \Aut \evolves{a,r} \Aut' $ 
is
 the smallest multi-relation   that satisfies the rules in Table \ref{simple-proc-sem}.
 
 We write $ \Aut \evolves{a,r} \Aut' $  if $(\Aut,(a,r),\Aut') \in \rightarrow$,  and  $\rightarrow^*$ for the  transitive closure of $\rightarrow$.

 %$\Aut \evolves{a,r}^* \Aut' $ if  $(\Aut,(a,r),\Aut')$ belongs to the transitive closure of $\rightarrow$.
\end{definition}
In what follows, we consider a  relation that is generally defined in  $\pi$-calculus \cite{Milner'99} {\em structural congruence}. Structural congruence is a relation  preserved by all operators of the calculus,  i.e a congruence, that identifies terms that should not be distinguished for semantical reasons. Hillston  \cite{Hillston94}  defined a similar relation in a operational  way  as  {\em  isomorphism}. 
\begin{definition} Structural  congruence, written $\equiv$ over the set of simple processes  $\Proc$ is the  smallest congruence that allows the reorder of terms in the choice.
\end{definition}
We use structural congruence as a relation to talk about individual terms  in the summation, and to to avoid to  be bothered by the order of  terms in the summation.
For  example $(a,\lambda).\Expr_1 +(b,\mu).\Expr_2 \equiv (b,\mu).\Expr_2+(a,\lambda).\Expr_1   $. We use $S$ to indicate a set  of terms of the summation that  we do not wish to identfy i.e. 
$(a,\lambda).\Expr_1 +(b,\mu).\Expr_2+ (c,\gamma).\Expr_3 \equiv (a,\lambda).\Expr_1 +S  $.

  Differently from \CTMC, not all transitions in SSPA  have a real number attached. These are called {\em passive transitions}. Passive transitions are transitions whose delay has not yet been specified. 
  The difference between  passive and active transitions, is crucial  in this work, so we proceed now to define such  entities via analysis of the labels in a process.
  \begin{definition}
A label   $a \in \Act$ is called  {\em  active }  with respect to  a simple process $\MA$ if $\MA \equiv (a,\lambda).\MA' +S$.
A label  $a \in \Act$ is called  {\em  passive}  with respect to a simple process  $\MA$ if  $\MA \equiv (a,x_a).\MA'+S$.
\end{definition}
We now define a the set of labels  that are active or passive in any possible evolution of the simple process.
\begin{definition}
The  set of {\em active labels},written $\actypeact{\MA}$,  is recursively defined as follows:
\begin{description}
\item[(Nil)] $\actypeact{\Nil}=\emptyset$;
\item[(Def)]$\actypeact{\Idf}=\actypeact{\MA}$\quad  if $\Idf\Defeq \MA$;
\item[(Choice)] $\actypeact{\sum_{i \in I}(a_i,r_i). \Expr_i}=\cup_{i\in I} \myset{a_i : (a_i,r_i).E_i, \, r_i \in \Real^+} \cup \actypeact{\Expr_i}$;
\item[(Closure)] $\actypeact{\clos{\Aut}{a}{\lambda}}= \actypeact{\MA}\cup \myset{a} $.
\end{description}
\end{definition}
\begin{definition}
The  set of {\em passive labels},written $\actypepas{\MA}$,  is recursively defined as follows:
\begin{description}
\item[(Nil)] $\actypepas{\Nil}=\emptyset$;
\item[(Def)]$\actypepas{\Idf}=\actypepas{\MA}$\quad  if $\Idf\Defeq \MA$;
\item[(Choice)] $\actypepas{\sum_{i \in I}(a_i,r_i). \Expr_i}=\cup_{i\in I} \myset{a_i: (a_i,r_i).E_i,\, r_i \in Var} \cup \actypepas{\Expr_i}$;
\item[(Closure)] $\actypepas{\clos{\Aut}{a}{\lambda}}= \actypepas{\MA}\backslash \myset{a} $.
\end{description}
\end{definition}
We simply  write $\set P$ and $\set A$ for the set of  passive and active labels when  it is clear from the context which simple process we are referring to.

\begin{definition} \label{close-process}
A simple process  $\Aut$  is {\em closed} if $\actypeact{\MA} =\emptyset$, it is open otherwise. 
 \end{definition}

The closure operator  transforms each open  automata into a closed one. 
We now present a few properties of the closure operator, with  respect to the semantics equivalence of {\em strong bisimilarity}.
\begin{definition}\label{str-bis}
We  define {\em strong bisimilarity} as the largest  symmetrical  relation $\cong$ such that   if $\Aut_1 \cong \Aut_2$, then  if  for all   $\Aut'_1$ it holds that $\Aut_1 \evolves{a,r} \Aut'_1  $ then  there exists $\Aut'_2$  such that  $\Aut_2 \evolves{a,r} \Aut'_2  $ and $\Aut'_1 \cong \Aut'_2$. 
\end{definition}
\begin{proposition}\label{simple-prop}
\begin{enumerate}
\item $\clos{\Aut}{a}{\lambda} \cong \MA $  if $\MA$ is closed. 
\item $\Aut_{[a \leftarrow \lambda][b \leftarrow \mu]}\cong\Aut_{b \leftarrow \mu][a \leftarrow \lambda]}$.
\item  Let $\actypepas{\MA}\cong\myset{a_1, a_2,\ldots a_N} $ be the set of passive actions of $\MA$.
$\MA_[a_1 \leftarrow\lambda_1]\ldots[a_N \leftarrow\lambda_N] $ is closed. 
\end{enumerate}
\end{proposition}
\begin{proof}
By  induction on grammar of simple processes.
\end{proof}
Sometimes, in the presence of multiple applications of the closure operator, we write $\MA_{ {\set  P} \leftarrow R}$, where $R$ is a set of rates $R= \myset{r_1,r_2,\ldots,r_N}$  and $ {\set  P}=\myset{a_1,a_2,\ldots,a_N}$ is the set of passive labels in $\MA$.  Clearly, by generalisation of Proposition \ref{simple-prop} this abbreviation  is well  defined, as it does not matter the order in which the closure is performed.
\begin{table}[h!]
\fbox{
\begin{minipage}{0.97\textwidth}
\[ \begin{array}{rl} 
\ruleOne{}{   \sum_{i\in I}(a_i,r_i).\Expr_i    \evolves{a_i,r_i} \Expr_i  }& 
\quad \quad \quad  \quad % \ruleOne{ (a,r).s_k \evolves{a,r} s_k}{\oplus^N_{i =1} s_i \evolves{a,r} s_k }\\
\ruleOne{\Expr \evolves{a,r} \Expr' }{ D \evolves{a,r} \Expr' }\quad
\mbox{ if } \,\, D \Defeq \Expr \\
\ruleOne{ (a,x_a).\Expr \evolves{a,x_a} \Expr}{ \clos{(a,x_a).\Expr}{a}{\lambda} \evolves{a,\lambda} \clos{\Expr}{a}{\lambda}  } &\quad \quad \quad

\ruleOne{(a,\mu).\Expr \evolves{a,\mu} \Expr}{ \clos{(a,\mu).\Expr}{a}{\lambda} \evolves{a,\mu} \clos{\Expr}{a}{\lambda}  }   \\
\end{array}\]
\end{minipage}
}
\caption{Transition semantics of simple processes}
\label{simple-proc-sem}
\end{table}
\begin{table}[h]
\fbox{
\begin{minipage}{0.97\textwidth}
\[\begin{array}{cc}
\ruleOne{\Aut_i \evolves{a,r} \Aut'_ i}{\coop{L}(\Aut_1,..,\Aut_i,.. \Aut_n)\evolves{a,r}\coop{L}(\Aut_1,..,\Aut'_i,.. \Aut_n) }& (a\notin L) \\
\ruleOne{\Aut_i \evolves{a,\lambda} \Aut'_i \qquad \qquad \Aut_k \evolves{a,x_a} \Aut'_k }{\coop{L}(\Aut_1,..,\Aut_i,..\Aut_k,.. \Aut_n)\evolves{a,\lambda}\coop{L}(\Aut_1,..,\Aut'_i,..,\Aut'_k,.. \Aut_n)  }& \,\, (a\in L,
								k\neq i) \\
\end{array}\]
\end{minipage}
}
\caption{Transition semantics of interacting processes}
\label{proc-sem}

\end{table}
Now we    define the  interaction among simple processes.
\begin{definition}\label{n-interactingautomata}    The set of { \em interacting processes},  $\Proctwo$,  is defined by the following syntax:
\[  \IAut ::=  \Aut ~\mid~\coop{L}(\Aut_1,\ldots ,\Aut_N) \]
 where  $L \subseteq  \Act$, $\Aut$ was defined in  Definition \ref{def:sp},  and    for all $i,j \leq N$ if  $i\neq j$ it holds that  $\actact(\Aut_i) \cap \actact(\Aut_j) \cap L= \pasact(\Aut_i) \cap \pasact(\Aut_j) \cap L= \emptyset$. %with $i\neq j$ and 
% $\Aut_i$ was defined in  Definition \ref{def:sp}.
\end{definition}
\begin{definition}\label{lts-coop} The   labelled transition system for the interacting processes  written $ \rightarrow:\Proctwo \times (\Act
\times \Real^+\cup \Var)
\times \Proctwo$  %written $ \IAut \evolves{a,r} \IAut' $
 is
 the smallest multi-relation   that satisfies the rules in Table \ref{simple-proc-sem}.

We write  $ \IAut \evolves{a,r} \IAut' $ if $ ( \IAut ,(a,r),\IAut') \in \rightarrow$,  and  $\rightarrow^*$  for the transitive closure of $\rightarrow$.
% $ \IAut \evolves{a,r}^* \IAut' $  if $ ( \IAut ,(a,r),\IAut')$  belongs to the  transitive closure of  .
\end{definition}

Crucial to this definition  is that the interaction  happens  pairwise.  For example, in queueing networks such as  the Jackson network \cite{Chao-Masa-Pinedo}, this captures the idea that customers hop from one node to one other.
\subsubsection{Comparison with PEPA}
The semantics  of the interaction  in  SSPA has been  inspired by PEPA \cite{Hillston94}, but it is not  identical. 
 PEPA's interaction operation  works on  {\em broadcasting} while in SSPA the  interaction/cooperation is strictly  pairwise. 

 In PEPA, cooperating processes over the same set of actions $L$ is  commutative and associative with the respect to a notion  of strong bisimulation. Strong bisimilarity  ($\cong$) identifies processes that can carry  out the same transitions with respect to the transition relation  defined in Definition \ref{lts-coop}. Therefore,   we assume that Definition  \ref{str-bis} is adapted to  the interacting processes. %  We  define strong bisimilarity as the largest  symmetrical  relation $\cong$ that satisfies the following condition: 
 % if  for all   $\IAut'_1$ it holds that $\IAut_1 \evolves{a,r} \IAut'_1  $ then  there exists $\IAut'_2$  such that  $\IAut_2 \evolves{a,r} \IAut'_2  $ and $\IAut'_1 \cong \IAut'_2$. 

 In SSPA, the cooperating  operator is commutative, but {\em not} associative with  respect to strong bisimilarity, even  under the same set of cooperating actions.
 Commutativity says that  the order in which processes cooperate  does not matter.  If fact,  two processes that differ only for the order of simple processes in the cooperation are strongly bisimilar,  and, we will see, they have the same  product-form  solution. However, as far associativity goes, the reader can verify that $(\MA_1 \oplus_L \MA_2)\oplus_L \MA_3$   and   $\MA_1 \oplus_L (\MA_2\oplus_L \MA_3)$ have different   transitions i.e they are not strongly bisimilar. %where the well-formed automata are defined as 
 To see this it suffices to take the following processes
 $(\MA_1= (a,\lambda).0$   and
  $\MA_2=\MA_3=(a,x_a).0$ with  $L=\myset{a}$ and verify that  $(\MA_2\oplus_L \MA_3)$  has no transition according to the semantics of SSPA. Therefore $\MA_1 \oplus_L (\MA_2\oplus_L \MA_3) \cong 0$ while $(\MA_1 \oplus_L \MA_2)\oplus_L \MA_3 \not \cong 0$, which  implies, differently  from PEPA semantics, that $(\MA_1 \oplus_L \MA_2)\oplus_L \MA_3 \not \cong\MA_1 \oplus_L (\MA_2\oplus_L \MA_3)$.  If we had used PEPA transition system we would have been able to  prove that the two processes are strongly  bisimilar.
\subsection{CTMC}
%\subsection{The name of the game}

In this paper,  we deal only with  product-form  solutions for {\em time-homogenous } {\CTMC}s. % In particular, we consider only   {\CTMC}s. 
For a  time-homogenous {\CTMC},  the  generator matrix $\generator$ contains all the  information to compute the  {\em transient} and {\em steady-state} probability  distribution.  
From the matrix $\generator$  it is possible to   describe the state space  of the {\CTMC} and vice versa.
Generally,  in  process algebra  such as PEPA,  the  language is  a means to  describe in a modular way the state space of the   underlying  {\CTMC}. The generator matrix  is then appropriately  recovered for computation  purposes.
If the interacting process $\IAut$ does not contain  passive transitions,  we can  recover the {\CTMC} by taking the set of all  derivatives of  $\IAut$ as the state space of the  {\CTMC }, and by generating the entries  of the generator matrix $\generator_ \IAut $ as the sum of all the real numbers of the transitions between two  derivatives.  Rates of self-loops  should be ignored, and the  diagonal of the generator matrix $\generator_ \IAut $  is constructed as usual to  ensure that the sum of the entries of the rows equals $0$. Even in the  presence of passive transitions in a process, the generator matrix  can  be recovered. However,  the matrix may not be used for computation purposes,  as it may  contain  the variables from the passive transitions. 
 For this reason, in what follows,  we describe how to  build the generator matrix, and we will  leave it to the  reader, or to the context in which it is used,  to  establish  if the generator can be used straightforwardly   for computation purposes, or  instantiation of variables is necessary.
 
 Given a process, we  define the {\em set of derivatives} as  the set of processes derived via the transitive closure of the labelled  transition  system. 

%The process underlying a closed process  is a {\CTMC} whose set of states  is the  multiset of derivatives of the interactive automata. 
The set of derivatives of an  interactive process  $\IAut$ is defined as  $\SetDerivatives_\IAut = \myset{\IAut' : \IAut \evolves{a,r}^*\IAut' }$.
The  transition rate from the state of the chain $\IAut$ to $\IAut'$ is given by the sum of the rates of
all the labelled  transitions of the process  i.e.:
\[
q(\transition{\IAut}{\IAut'}{})=\sum_{\substack{(a,r): \IAut 
     \evolves{a,r} \IAut' \\ \IAut\neq \IAut'}} r. 
 \]
 
If  all transition rates $q(\transition{\IAut}{\IAut'}{}) \in \Real^+$ then $\generator_{\IAut}$   is the generator matrix of the  underlying  {\CTMC} of the interactive process.
If  for  some rates it holds that    $q(\transition{\IAut}{\IAut'}{}) \not \in \Real^+$ then  we must specify that  
the variables in the prefix $(a,x).\Expr $ are considered the same if they occur with  the same label.  This  observation  has a huge impact in correct derivation  of the generator matrix.  For all  intents and purposes,  variables are considered the same if they are associated with the same label. For  example we could write $ \Aut=(a,x).(b,\mu).\MA  +(a,y).(c,z).\MA$, however,  in the construction of the generator matrix,  either $x$ or $y$ will appear in the definition of the rate.  This concept has no meaning from the point of view of the process algebra, but it has huge impact in the building of the generator matrix, and in the computation of  probabilities.     
The  generator matrix of $\Aut$, written  $\generator_{\Aut}$, will  be
\[ \generator_{\Aut}=\left( \begin{array}{ccc}
-2x&x &x \\
 	\mu& -\mu & 0\\
	0&z& -z\end{array}
\right) \]
while the matrix 
\[\generator_{\Aut}=\left( \begin{array}{ccc}
  -(x+y)&x&y \\
 	\mu& -\mu & 0\\
	0&z& -z
\end{array}
\right) \]
is not what we intended.
We impose that $x=y$, since they appear with the label $a$  and therefore  we can treat $x$ as a variable, and apply standard numerical  operations. The semantics for the variables is such that $x,y\neq z$ as $z$ occurs with the label $c$, not $a$.
For this reason  the subscript of the label of the action in variables is  used in this paper.

For convenience we write $q(\transition{\IAut}{\IAut'}{a})=\sum_{r: \IAut \evolves{a,r} \IAut' } r$ 
for the transition rate with  respect to a label $a$.   We note that for the latter   definition  we also consider   rates from a state to itself i.e. 
$q(\transition{\IAut}{\IAut}{a})$. This  will  be useful later in Theorem \ref{grcat}.
We observe that $q(\transition{\IAut}{\IAut'}{}) = \sum_{\substack{a \in \Act  \\ \IAut\neq \IAut'}}q(\transition{\IAut}{\IAut'}{a})$. 

Given an  interacting process $\IAut$, we can  generate the state space of the {\CTMC} as $\SetDerivatives_\IAut$ and  the generator matrix $\generator_{\IAut}$, then 
 with an abuse of notation  we refer to  the properties
of the process meaning the properties of  the {\CTMC}.
Therefore, we can talk about a stationary or  steady-state distribution of the process, $\pi_\IAut$, meaning that its
{\CTMC} has a stationary or steady-state distribution  $\pi$. 
If $\generator_\IAut$ is the generator matrix of the underlying {\CTMC} of $\IAut$, 
  then we write  $\pi(\IAut)$  for {\em invariant measure} meaning the $\pi\generator_{\IAut}={\mathbf 0}$. In other words, we use in the notation 
  $\pi(\IAut)$ instead of  $\pi\generator_{\IAut}$.
  For  $\IAut' \in  \SetDerivatives_\IAut $ we write $\vec{\pi}_{\MA}(\IAut') $ meaning the value of the vector  $\vec{\pi}_\MA$ for the element $\IAut'$.
If $\sum_{\IAut' \in \SetDerivatives_\IAut} \pi(\IAut')=1$ then the {\CTMC} is ergodic and $\pi$ is the state-state distribution of $\IAut $
\cite{Chao-Masa-Pinedo}.

 \section{Product-form  solution}

We now present the main  theorem  of the paper regarding product-form solution for SSPA. The theorem asserts that for a given class of processes,  that satisfies  certain conditions on the rates and on the structure of state space,  the product-form solution  exists. 
We start with the definition  of structure of the processes.
\begin{definition}
 In a simple process $\MA =\sum_{i \in I}(a_i,r_i). \Expr_i$ the label $a$ is the {\em unique passive label}  if and only if $ \MA\equiv (a,x_a).  \Expr +S $ and $ \MA\equiv (a,x_a).  \Expr' +S' $  then $S'=S$ and $\Expr =\Expr '$ .
 \end{definition}
 We now define a  set of unique passive labels for a process.  Such  a set is not empty  if in all possible evolution of the process, one passive transition with a given  label is possible.
\begin{definition} The set of {\em unique passive labels}  in a simple process $\MA$, written $\uniqueLabels{\MA}$, is recursively defined as follows:
\begin{description}
\item[(Nil)] $\uniqueLabels{\Nil}= \Act$;
\item[(Def)]$\uniqueLabels{\Idf}=\uniqueLabels{\MA}$ \quad  if $\Idf\Defeq \MA$ 
\item[(Choice)] $\uniqueLabels{\sum_{i \in I}(a_i,r_i). \Expr_i} =  \left\{ \begin{array}{ll} 
								\emptyset	& \mbox{if there exist a passive}\\
								&\mbox{ label in}   \sum_{i \in I}(a_i,r_i) .\Expr_i \\
									&  \mbox{ which  is not unique}  \\ 
									&\\
								    (\cup_{i\in I}\myset{a_i})\cap_{i\in I} \uniqueLabels{ \Expr_i}     & \mbox{\quad if } a_i \mbox{is a unique } \\
								    &\quad \mbox{passive label  in} \\
								    &\quad   \sum_{i \in I}(a_i,r_i) .\Expr_i   \\
								\end{array} \right. $

\item[(Closure)] $\uniqueLabels{\clos{\Aut}{a}{\lambda}}= \uniqueLabels{\MA} \backslash\myset{a} $
\end{description}
\end{definition}
Such a definition is necessary as we need to use process that have one passive transition. This restriction could be relaxed, but it would involve a more  complicated  statement of Theorem  \ref{grcat}.  

We now provide the definition of well-formed simple processes, which are the building blocks for the correct definition of product form solutions. Well-formed processes are processes that will  generate no confusion in the construction of product form solutions.  Informally,  we can  think product-form as a way  of  decomposing the invariant measure of a {\CTMC}.  
Now, if  the {\CTMC} has been  built using simple processes and an empty cooperation set, then  each simple process is independent of the other, and trivially the invariant measure can be written as the product of the invariant measures of each  simple process. However, if a complex  {\CTMC} has been  built using simple processes and a {\em non-empty} cooperation, 
then the behaviour of  each simple process can  be influenced by the others in the cooperation.  If, however, in each simple process, the reversed rates of the  cooperating transitions are constant in each state,    then the  invariant measure  of a complex   {\CTMC} can  be written as the product of the invariant measures of each  simple process, in a similar fashion as if   they  were independent. 
To perform  all these calculations correctly, we need to make sure that no confusion arises when  writing the processes in SSPA, and therefore we need the notion of {\em well-formed processes}.
 
\begin{definition}[Well-formed processes]\label{well-formed}  A  simple  process  $\MA$ is {\em well-formed} if:
\begin{enumerate}
  \item  \label{fcond}
     $ \actypeact{\MA} \cap \actypepas{\MA}=\emptyset $ and 
    \item \label{scond} %if $a \in \actypepas{\MA}$, then $a\in \uniqueLabels{\MA}$ i.e 
     if  $\actypepas{\MA}\neq \emptyset$  then  $\actypepas{\MA}=\uniqueLabels{\MA}$. 
    %   \[\begin{array}{c}  \forall s \exists s'  \in \statespace  \quad \mbox{ such that } \quad s \evolves{a,x_a} s' \\
    %l\forall s, s', s'' \in \statespace,\ s \evolves{a,x_a} s' \land s \evolves{a,x_a} s'' \implies s' = s''
    %\end{array}\]
  \end{enumerate}
  \end{definition}

From a syntactic point of view, we have done the work  for the following result for the product-form solution. The theorem considers only  complex {\CTMC}s composed by well-formed processes.   A further condition is added on the  outgoing  rates of the simple processes to guarantee that  on average, we can  quantify the dependency among the various processes.

\begin{theorem} \label{grcat}
Given an  interacting process $\IAut=\bigoplus_{L}(\MA_1,\MA_2,\ldots ,\MA_N) $ composed  by be  well-formed simple processes    $\MA_1,\MA_2,\ldots ,\MA_N$ that cooperate on a finite  set of actions 
$L=\{a_1,a_2, \ldots a_M\}$.

Assume that the state space of  $\SetDerivatives_\IAut=\SetDerivatives_{\MA_1}\times \SetDerivatives_{\MA_2} \times \dots \times\SetDerivatives_{\MA_N} $ is irreducible. %   $S_1 \times S_2 \times  \ldots \times S_N $ is irreducible.
   %there exists the set of rates $\{\kappa_1, \ldots, \kappa_M\}$ which satisfie the following equations:
   If for all labels in the cooperation set $L$    there exists a set of positive real numbers  
   $K=\{\kappa_1, \ldots, \kappa_M\}$ such that, for any simple process $\MA_i$, the following equations are satisfied
  \begin{equation}\label{eq:main.condition}
        \frac{ \sum_{\MA'\in \SetDerivatives_{\MA_i}} q(\transition{\MA'}{\MA}{a_r}) \pi_i(\MA')}{\pi_i(\MA)}=\kappa_r % i,j\in[1,\ldots,N], \,\,
       \quad  \quad   \quad   \quad    \quad   \MA \in \SetDerivatives_{\MA_i}, \,\, a_r \in L\cap  \actact(\MA_i)
  \end{equation}
  where $\pi_i$ is  the invariant measure of      the closed process $\MA_i^c=\MA_{i[{\set  P}\cap L \leftarrow K] }$.
Then  the following statements hold:
   \begin{enumerate}
   \item The invariant measure of the process $\bigoplus_{L}(\MA_1,\MA_2, \ldots, \MA_N)$  has the  product-form:
  \begin{eqnarray}
  \pi(\bigoplus_{L}(\MA_1,\MA_2, \ldots, \MA_N)) =  \pi_1(\MA^c_1)\otimes \pi_2(\MA^c_2) \otimes\ldots\otimes \pi_2(\MA^c_N) \label{eq:product.form}   \end{eqnarray}
  where $\otimes $ is the  Kronecker product \footnote{If $\pi_1,\pi_2$ are  two vectors, $\pi_1\in \Real^{1\times n}$ and $\pi_2\in \Real^{1\times m}$ then he  Kronecker product  is   $\pi_1\otimes \pi_2 =(p_1\pi_2,p_2\pi_2,\dots p_n\pi_2) \in\Real^{1\times nm} $. }.
  \item If $\sum_{\MA \in \statespace_{\MA_i}} \pi_i(\MA) = 1$  $(i \in [1,\ldots,N])$  then  $\pi$ is 
    the steady-state probability distribution of $\bigoplus_{L}(\MA_1,\MA_2, \ldots, \MA_N)$.
  \end{enumerate}

    \end{theorem}
    \begin{proof} %See Appendix 1.
     For (1) we first show that for all states $(\MA_1,\MA_2, \ldots, \MA_N) \in \SetDerivatives_{\MA_1}\times \SetDerivatives_{\MA_2} \times \dots \times\SetDerivatives_{\MA_N}$ we can  derive  $\pi( \MA_1,\MA_2, \ldots, \MA_N) =\prod^N_{i=1} \pi_i(\MA_i)$ for $ \MA_i\in\SetDerivatives_{\MA^c_i}$.
    Since we are considering the whole state space, the result  $  \pi(\bigoplus_{L}(\MA_1,\MA_2, \ldots, \MA_N)) =  \pi_1(\MA^c_1)\otimes \pi_2(\MA^c_2) \otimes\ldots\otimes \pi_2(\MA^c_N)$ follows.
   We show now, for $N=2$ that $\pi( \MA_1,\MA_2)= \pi_1(\MA_1) \pi_2(\MA_2)$ when $ \MA_1 \oplus_{\myset{a,c} \MA_2}$. Generalisation to $N$ is straightforward.
   We observe, that with an  abuse of notation we write $\MA_1$ to indicate the simple process in the cooperation, but also  the process that forms the state space of  $\stsp{ \MA_1}$. The context distinguishes between these two mathematical  objects. 
   
    The global  balance equations  for the process $\MA_1$  or  $\MA^c_1$ are the following, assuming that $\actypeact{\MA_1}\cap L = \myset{a}$ and $\actypeact{\MA_2}\cap L = \myset{c}$
    \begin{multline}
\pi_{\MA^c_1}(\MA_1)\big(  \sum_{\MA_1'\in \stsp{ \MA^c_1}} \rateact{\MA_1}{\MA_1'}{a}{ \MA^c_1} + \underbrace{ \rateact{\MA_1}{\MA_1'}{c}{ \MA^c_1}  }_{x_c} + \sum_{\substack{\MA_1'\in \stsp{ \MA^c_1} \\ b\neq a,c}} \rateact{\MA_1}{\MA_1'}{b}{ \MA^c_1}\big) 
=  \\ \sum_{\substack{\MA_1' \in \stsp{ \MA^c_1} }}\rateact{\MA_1'}{\MA_1}{a}{ \MA^c_1}\pi_{ \MA^c_1}(\MA_1') +\sum_{\substack{\MA_1' \in \stsp{ \MA^c_1} }}\underbrace{\rateact{\MA_1'}{\MA_1}{c}{ \MA^c_1}}_{x_c}\pi_{ \MA^c_1}(\MA_1')  +\\
 \sum_{\substack{\MA_1'\in \stsp{ \MA^c_1}\\b\neq a,c}}  \rateact{\MA_1'}{\MA_1}{b}{ \MA^c_1}\pi_{ \MA^c_1}(\MA_1').  \label{gbeP}
\end{multline}
We have underlined the transition  rates what would have a variable in  $\MA_1$ but a real  number in  $\MA^c_1$. By  definition of well-formed simple process, we know that there is only  one instance of $\rateact{\MA_1}{\MA_1'}{c}{ \MA^c_1}$.
 
For $\MA_2$ or $\MA^c_2$ the global  balance equations would be similar by  reverting the role of the rates of  the actions $a$ and $c$.% and substituting the correct states.

%%% EQUATION 1
We write now the global balance equations for the global state $(\MA_1,\MA_2)$ as follows:
\begin{multline}
\pi\big((\MA_1,\MA_2)\big)\Big(\sum_{\substack{\MA_1'\in \stsp{ \MA_1}\\ b\neq a,c}}\rateact{(\MA_1,\MA_2)}{(\MA'_1,\MA_2)}{b}{}+\sum_{\substack{\MA_2'\in \stsp{\MA_2}\\ b\neq a,c}} \rateact{(\MA_1,\MA_2)}{(\MA_1,\MA_2')}{b}{}  + \\ 
 \sum_{(\MA'_1,\MA_2')\in   \stsp{ \MA_1} \times \stsp{\MA_2}}\rateact{(\MA_1,\MA_2)}{(\MA'_1,\MA_2')}{c}{}+ \sum_{ (\MA'_1,\MA_2')\in   \stsp{ \MA_1} \times \stsp{\MA_2}}\rateact{(\MA_1,\MA_2)}{(\MA'_1,\MA_2')}{a}{} \Big) \\
 %%% EQUALS
=  \sum_{\substack{\MA_1'\in \stsp{\MA_1}\\ b \neq a,c}}  \rateact{(\MA'_1,\MA_2)}{(\MA_1,\MA_2)}{b}{}\pi\big((\MA'_1,\MA_2)\big) +  \\ \sum_{\substack{\MA_2'\in \stsp{\MA_2}\\ b\neq a,c}}  \rateact{(\MA_1,\MA_2')}{(\MA_1,\MA_2)}{b}{}\pi\big((\MA_1,\MA_2')\big)  \, + \\
\sum_{(\MA'_1,\MA_2')\in   \stsp{ \MA_1} \times \stsp{\MA_2}}\rateact{(\MA'_1,\MA_2')}{(\MA_1,\MA_2)}{a}{}\pi\big((\MA'_1,\MA_2') +  \nonumber
\end{multline} 
\begin{multline} 
\sum_{(\MA'_1,\MA_2')\in   \stsp{ \MA_1} \times \stsp{\MA_2}}\rateact{(\MA'_1,\MA_2')}{(\MA_1,\MA_2)}{c}{}\pi\big((\MA'_1,\MA_2')\big). \nonumber
\end{multline} 

 We assume that we can write the joint invariant measure in  product-form, dividing by  $\pi_{\MA_1}(\MA_1), \pi_{\MA_2}(\MA_2)$  and  writing down  the contribution  of the rates of  each simple process for the labels $b \notin L$ we have:

\begin{multline}
\sum_{\substack{\MA_1'\in \stsp{ \MA_1}\\ b\neq a,c}}\rateact{\MA_1}{\MA'_1}{b}{}+ \sum_{(\MA'_1,\MA_2')\in   \stsp{ \MA_1} \times \stsp{\MA_2}}\rateact{(\MA_1,\MA_2)}{(\MA'_1,\MA_2')}{c}{}  + \\
\sum_{\substack{\MA_2'\in \stsp{\MA_2}\\ b\neq a,c}} \rateact{\MA_2}{\MA_2'}{b}{} + \sum_{(\MA'_1,\MA_2')\in   \stsp{ \MA_1} \times \stsp{\MA_2}}\rateact{(\MA_1,\MA_2)}{(\MA'_1,\MA_2')}{a}{}  \\
 %%% EQUALS
=  \sum_{\substack{\MA_1'\in \stsp{\MA_1}\\ b \neq a}}  \rateact{\MA'_1}{\MA_1}{b}{}\frac{\pi_{\MA_1}(\MA_1')}{ \pi_{\MA_1}(\MA_1)} + 
  \sum_{\substack{\MA_2'\in \stsp{\MA_2}\\ b\neq a,c}}  \rateact{\MA_2'}{\MA_2}{b}{}\frac{\pi_{\MA_2}(\MA_2')}{\pi_{\MA_2}(\MA_2)}  \, + \\
\sum_{(\MA'_1,\MA_2')\in   \stsp{ \MA_1} \times \stsp{\MA_2}}\rateact{(\MA'_1,\MA_2')}{(\MA_1,\MA_2)}{a}{}\frac{\pi_{\MA_1}(\MA_1')\pi_{\MA_2}(\MA_2')}{\pi_{\MA_1}(\MA_1)\pi_{\MA_2}(\MA_2) } 
\, + \\
 \sum_{(\MA'_1,\MA_2')\in   \stsp{ \MA_1} \times \stsp{\MA_2}}\rateact{(\MA'_1,\MA_2')}{(\MA_1,\MA_2)}{c}{}\frac{\pi_{\MA_1}(\MA_1')\pi_{\MA_2}(\MA_2')}{\pi_{\MA_1}(\MA_1)\pi_{\MA_2}(\MA_2) }.
 %\pi_{\MA_1}(\MA_1')\pi_{\MA_2}(\MA_2')
 \nonumber 
\end{multline} 
We consider the  rates in the terms with joint  states. We observe that  since we impose that the simple processes are well  formed, this means that  there is only  one passive action in  each process:  in $\MA_1$ the passive action will  be labelled $c$ while in  $\MA_2$ will  be labelled $a$. The number of transitions  in the joint state space $\stsp{ \MA_1} \times \stsp{\MA_2} $ will  the same number as the active transitions. Therefore we can  rewrite the global  balance equation as follows:

\begin{multline}
\sum_{\substack{\MA_1'\in \stsp{ \MA_1}\\ b\neq a,c}}\rateact{\MA_1}{\MA'_1}{b}{}+\sum_{\substack{\MA_2'\in \stsp{\MA_2}\\ b\neq a,c}} \rateact{\MA_2}{\MA_2'}{b}{}  + \sum_{\MA_2'\in  \stsp{\MA_2}}\rateact{\MA_2}{\MA_2'}{c}{} +  \\ \sum_{\MA'_1\in   \stsp{ \MA_1}}\rateact{\MA_1}{\MA'_1}{a}{}  
 %%% EQUALS
=  \sum_{\substack{\MA_1'\in \stsp{\MA_1}\\ b \neq a}}  \rateact{\MA'_1}{\MA_1}{b}{}\frac{\pi_{\MA_1}(\MA_1')}{ \pi_{\MA_1}(\MA_1)} + \\ 
  \sum_{\substack{\MA_2'\in \stsp{\MA_2}\\ b\neq a,c}}  \rateact{\MA_2'}{\MA_2}{b}{}\frac{\pi_{\MA_2}(\MA_2')}{\pi_{\MA_2}(\MA_2)}  \, + 
 \sum_{\MA_2'\in  \stsp{\MA_2}}\sum_{\MA'_1\in   \stsp{ \MA_1} }\rateact{\MA'_1}{\MA_1}{a}{}\frac{\pi_{\MA_1}(\MA_1')\pi_{\MA_2}(\MA_2')}{\pi_{\MA_1}(\MA_1)\pi_{\MA_2}(\MA_2) } \\
+ \sum_{\MA_1'\in  \stsp{\MA_1}} \sum_{\MA_2'\in  \stsp{\MA_2}}\rateact{\MA_2'}{\MA_2}{c}{}\frac{\pi_{\MA_1}(\MA_1')\pi_{\MA_2}(\MA_2')}{\pi_{\MA_1}(\MA_1)\pi_{\MA_2}(\MA_2) }.
 %\pi_{\MA_1}(\MA_1')\pi_{\MA_2}(\MA_2')
 \nonumber 
\end{multline} 
We can  now rewrite the  global balance equations in  Equation \ref{gbeP} in the following convenient way:

  \begin{multline}
  \sum_{\MA_1'\in \stsp{ \MA^c_1}} \rateact{\MA_1}{\MA_1'}{a}{ \MA^c_1} +  \kappa_c  + \sum_{\substack{\MA_1'\in \stsp{ \MA^c_1} \\ b\neq a,c}} \rateact{\MA_1}{\MA_1'}{b}{ \MA^c_1}
= \\ \kappa_a+   \sum_{\substack{\MA_1' \in \stsp{ \MA^c_1} }}\kappa_c \frac{\pi_{ \MA^c_1}(\MA_1') }{\pi_{\MA^c_1}(\MA_1)} +
 \sum_{\substack{\MA_1'\in \stsp{ \MA^c_1}\\b\neq a,c}}  \rateact{\MA_1'}{\MA_1}{b}{ \MA^c_1}\frac{\pi_{ \MA^c_1}(\MA_1') }{\pi_{\MA^c_1}(\MA_1)}\big) \nonumber  
 \end{multline}

By subtracting each  term  side of the last two  equation we obtain:
\begin{multline}
\sum_{\substack{\MA_2'\in \stsp{\MA_2}\\ b\neq a,c}} \rateact{\MA_2}{\MA_2'}{b}{}  + \sum_{\MA_2'\in  \stsp{\MA_2}}\rateact{\MA_2}{\MA_2'}{c}{}   -\kappa_c 
 %%% EQUALS
=  \\ 
  \sum_{\substack{\MA_2'\in \stsp{\MA_2}\\ b\neq a,c}}  \rateact{\MA_2'}{\MA_2}{b}{}\frac{\pi_{\MA_2}(\MA_2')}{\pi_{\MA_2}(\MA_2)}  \, + 
 \sum_{\MA_2'\in  \stsp{\MA_2}}\kappa_a\frac{\pi_{\MA_2}(\MA_2')}{\pi_{\MA_2}(\MA_2) }- \kappa_a
 %\pi_{\MA_1}(\MA_1')\pi_{\MA_2}(\MA_2')
 \nonumber 
\end{multline} 
The latter equation can  be rewritten  to see that the global balance equation of $\MA_2$ by expanding the Definition  of  $ \kappa_a, \kappa_c$ as in Condition  \ref{eq:main.condition} of Theorem \ref{grcat}.
    \end{proof}
The proof is very  elegant, not because it uses global balance equations, but because it solidly relies on the semantics of the cooperation. Such semantics  establishes  the contribution of each component  to transform the global balance equations  of the joint processes into the global balance equations of each simple process.
Further considerations  on the cooperation operator will lead to conclude that  the semantics given  in this work  is the right one, as it allows the correct substitution of  the rates in condition \ref{eq:main.condition} of Theorem \ref{grcat} in the passive transitions.  As Hillston  pointed out in \cite{Hillston94}, passive transitions lack of information about the speed of the transition. Such  information is  given  by the cooperation with the active  partner. In  this work we embrace this view fully, but we also find out the right rates (the ones given  by condition \ref{eq:main.condition} of Theorem \ref{grcat})   for the passive transitions to proceed in isolation.  We could have chosen an arbitrary rate to be substituted into  the passive transition. This would have made no sense at all.  The product form  solution  relates  the rates of the joint process with the rates of the single components. We can see  why it is important that cooperation is not associative, and the broadcasting semantics of PEPA would not work  here. Consider three well-formed simple processes $\MA_1,\MA_2,\MA_3$ specified as follows $\MA_1=(a,\lambda).\MA'_1, \MA_2= (a,x_a).\MA'_2, \MA_3= (a,x_a).\MA'_3$ such that condition \ref{eq:main.condition} of Theorem \ref{grcat} is satisfied for $\MA_1$.  Assume that we have    PEPA  semantics and $ \MA_2 \oplus{\myset{a}}\MA_3  \evolves{a,x_a}  \MA'_2 \oplus{\myset{a}}\MA'_3 $ and $ \MA_1 \oplus{\myset{a}} (\MA_2 \oplus{\myset{a}}\MA_3)  \evolves{a,\lambda} \MA_1 \oplus{\myset{a}} (\MA'_2 \oplus{\myset{a}}\MA'_3)$. Now, to identify the product-form solution we would need  to substitute in  both processes  the rate  $\kappa_a$
 $ \MA_{2 [a \leftarrow \kappa_a]},\MA_{3 [a \leftarrow \kappa_a]}$.  The product from  would not work at all. The reader can verify this by inspecting the proof of theorem  \ref{grcat}.   Now consider the  equivalent process $ (\MA_1 \oplus{\myset{a}}\MA_2)\oplus{a}\MA_3 \evolves{a,\lambda}  (\MA'_1 \oplus{\myset{a}}\MA'_2)\oplus{\myset{a}}\MA'_3$  such that $ (\MA_1 \oplus{\myset{a}}\MA_2) \evolves{a,\lambda}  (\MA'_1 \oplus{\myset{a}}\MA'_2)$. We would obtain  a series of substitutions   $\MA_{2 [a \leftarrow \kappa_a]}$ and, if $(\MA_1 \oplus{\myset{a}}\MA_2)$ satisfy the conditions in  theorem \ref{grcat} for $a$, then  we would have    $\MA_{3 [a \leftarrow  \kappa^*_a]}$, where $\kappa^*_a= \sum_{ (\MA^*_1,\MA^*_2) } \lambda\frac{\pi(\MA^*_1,\MA^*_2) }{\pi(\MA''_1,\MA''_2) }$ Clearly this would lead to different product-form solution from $ \MA_1 \oplus{\myset{a}} (\MA_2 \oplus{\myset{a}}\MA_3)$.  
 In conclusion,  differently from PEPA semantics, we do  not wish to have associativity as the  rates used for  the closure  of each simple process matters. Such  rates depend on  how we group  simple processes together.
The semantics of the cooperation is the exactly was is needed to  correctly interpret  product-form solutions.
 In  this work  we are not concerned about  numerical  or analytical methods for solution  equations in the form  of \ref{eq:main.condition}. Such  methods can  be found in \cite{Chao-Masa-Pinedo}.
\section{Product-form  solutions for biological systems}
As stated in the introduction, product-form solutions have been  mostly used in queueing theory. There has been a recent interest in product-form solution for  biological system \cite{Mairesse09,Anderson10}. In particular in \cite{Mairesse09,Anderson10} consider only chemical reactions, while  we show here a variation of the product-form solutions for more complex systems.  

Assume we have a cancerous  cell, that grows, in a limited way  provided that there is enough energy. In absence of energy  the cell  could die, with a certain probability $p$. 
We model the cell  as follows:
\[\begin{array}{rcl}
\newExpr{C}_0&=&(a,x_a).\newExpr{C}_1 \\%+ (b,\mu).\newExpr{C}_0\\ 
\newExpr{C}_1&=&(a,x_a).\newExpr{C}_2+ (c,\gamma_1).\newExpr{C}_0 + (c,\kappa_c).\newExpr{C}_1\\ 
\vdots& &  \vdots \\
\newExpr{C}_N&=&(a,x_a).\newExpr{C}_N+  (c,\gamma_N).\newExpr{C}_0 + (c,\kappa_c).\newExpr{C}_N.\\ 
\end{array}\]
The transitions labelled $a$ stand for the energy that  will allow the  cell  to  grow. As energy  is provided by the environment, we model it as a passive transition.  We note that the cell has a finite growth: even  in the presence of an  infinite amount of energy, the cell will stop growing.
The transitions labelled $c$ stands for  cancerous events: they could inhibit  growth and kill the cell.
  
We model the energy as a switch, either there is energy  for the cell to  grow, or there is no energy:

\[\begin{array}{rcl}
\newExpr{E}_0&=&(a,\lambda).\newExpr{E}_1+ (a,\delta).\newExpr{E}_0\\
\newExpr{E}_1&=&(d,\delta).\newExpr{E}_0.
\end{array}\]
The rate $\delta$ represents the  speed at which the environment supplies energy. 
Similarly, we model the trigger for cancer which  reduces the  size of the cell as a switch:
\[\begin{array}{rcl}
\newExpr{T}_0&=&(c,x_c).\newExpr{T}_1 + (c,x_c).\newExpr{T}_0\\
\newExpr{T}_1&=&(e,\nu).\newExpr{T}_0.
\end{array}\]
f
The system is the following: ${\oplus}_{\myset{a,c}}(\newExpr{E}_0 ,\newExpr{ C}_0 , \newExpr{T}_0 )$, and the  steady state probabilities are $\pi({\oplus}_{a,c}(\newExpr{E}_0 ,\newExpr{C}_0 , \newExpr{T}_0 ) )=\pi_1(\newExpr{E}_0) \otimes \pi_2(\newExpr{ C}_{0 [a \leftarrow \delta]})\otimes \pi_3( \newExpr{T}_{0 [c \leftarrow \kappa_c]}) $ where
 $\kappa_c=\frac{ \gamma_3 \pi_1(\newExpr{C}_3) +  \gamma_2 \pi_1(\newExpr{C}_2)+ \gamma_1 \pi_1(\newExpr{ C}_1)}{\pi_1(\newExpr{ C}_0)} $. 
 We observe that  $\newExpr{E}_0 ,\newExpr{ C}_0 , \newExpr{T}_0 $ are well formed process and that the conditions of Theorem \ref{grcat} are satisfied.

The transition graphs  of  the state-space processes can  be found in Figure \ref{biological-systema}.

\begin{figure}[h]
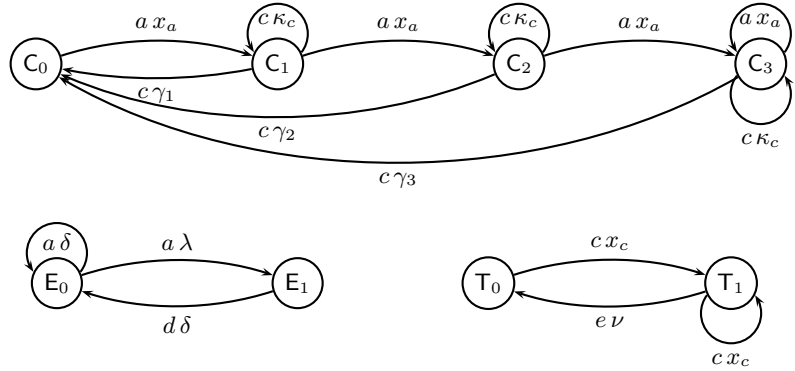

\centering
$\psmatrix[mnode=circle,colsep=2.5]  \newExpr{C}_0& \newExpr{C}_1 & \newExpr{C}_2&\newExpr{C}_3 
\endpsmatrix
\psset{shortput=nab,arrows=->,labelsep=3pt} \small
 \ncarc[arcangle=18]{1,1}{1,2}^{a\, x_a}
  \ncarc[arcangle=18]{1,2}{1,3}^{a\, x_a}
  \ncarc[arcangle=18]{1,3}{1,4}^{a\, x_a}
  \psset{arrows=->,labelsep=3pt} 
   \nccircle[arcangle=10]{1, 4}{0.4}^{a\, x_a}
   \nccircle[arcangle=10]{1, 4}{-0.4}^{c\, \kappa_c}
   \nccircle[arcangle=10]{1, 3}{0.4}^{c\, \kappa_c}
      \nccircle[arcangle=10]{1, 2}{0.4}^{c\, \kappa_c}
\ncarc[arcangle=30]{1,4}{1,1}^{c\, \gamma_3}
\ncarc[arcangle=23]{1,3}{1,1}^{c\, \gamma_2}
\ncarc[arcangle=11]{1,2}{1,1}^{c\, \gamma_1}
$

\vspace*{2cm}
\subfigure{
$ \begin{psmatrix}[mnode=circle,colsep=2.5]
\newExpr{E}_0& \newExpr{E}_1 \end{psmatrix}
 \psset{shortput=nab,arrows=->,labelsep=3pt} \small
 \ncarc[arcangle=18]{1,1}{1,2}^{a\, \lambda}
  \ncarc[arcangle=18]{1,2}{1,1}^{d\, \delta}
   \nccircle[arcangle=10]{1, 1}{0.4}^{a\, \delta}$
   }\hspace*{1.5cm}
   \subfigure{ 
  $ \begin{psmatrix}[mnode=circle,colsep=2.5]
\newExpr{T}_0& \newExpr{T}_1 \end{psmatrix}
 \psset{shortput=nab,arrows=->,labelsep=3pt} \small
 \ncarc[arcangle=18]{1,1}{1,2}^{c\, x_c}
  \ncarc[arcangle=18]{1,2}{1,1}^{e\, \nu}
   \nccircle[arcangle=10]{1, 2}{-0.4}^{c\, x_c}$
   }
%\ncarc[arcangle=25]{1,4}{1,2}^{c\, \gamma'_3}
%\ncarc[arcangle=15]{1,4}{1,3}^{c\, \gamma''_3}
%\
% \nccircle[arcangle=10]{1, 1}{0.4}^{c\, \kappa_c}$
\label{biological-systema}
\vspace*{1cm}
\caption{ Transition graphs  of the  state-space of the processes  $\newExpr{E}_0 ,\newExpr{ C}_0 , \newExpr{T}_0 $ }
\end{figure}
%\cite{Hillston-rep-pf98,sereno:pfpepa,boucherie:characterisation,hillston:product}\cite{harrison:exploiting,balbo:pfgspn,Kelly'79,Chao-Masa-Pinedo,Robertazzi}
%\cite{harrison:rcat,BHHK03,jackson:qn,muntz:mim1,bcmp,gelenbe:negative.customers}\cite{coleman:product.form,balbo:pfgspn,balbo:relationsbcmp-pfspn}
 \section{Conclusion}
 The main  interest for product-form solutions arises when  a  {\CTMC}   contains a rather large state space, and the analytical  computation of the steady state probability/invariant measure    can be  computationally   prohibitive.
In this paper, we have analysed  the cooperation  operator for a simple process algebra and its relationship with  product-form solution. We have clarified the semantics of such  operator,  and we have  shown that such semantics  is necessary   to derive  correctly product-form solutions.
%and to  understand the basic  principles that govern  product-form solutions. %\cite{Hillston-rep-pf98,sereno:pfpepa,boucherie:characterisation,hillston:product}\cite{harrison:exploiting,balbo:pfgspn,Kelly'79,Chao-Masa-Pinedo,Robertazzi}
%\cite{harrison:rcat,BHHK03,jackson:qn,muntz:mim1,bcmp,gelenbe:negative.customers}\cite{coleman:product.form,balbo:pfgspn,balbo:relationsbcmp-pfspn}

\section*{Acknowledgments}
I gratefully acknowledge Jane Hillston for useful discussions  on  product-form  solutions and operational  semantics. Jane pointed out some mistakes and typos in  earlier version of the proof of theorem \ref{grcat}.
A lot of   generous support and encouragement   was given by Gianfranco Balbo. I had very interesting  discussions  with Andrea Marin,  who also  pointed out to the  work  done by Mairesse on biological  systems. My  interest in  product-form solution has arisen  during research  work I conducted with  Peter Harrison. His  unique way  of working has been  great inspiration for me. 
%\section{Conclusion}
 
\bibliographystyle{plain}
 
 \bibliography{newbibliography}
%%\appendix{Measure Theory}
%%\appendix{Topology}

\end{document}